\documentclass[12pt]{article}

\usepackage{amssymb}
\usepackage{latexsym}
\usepackage{amsmath}
\usepackage{hyperref}
\usepackage{amsthm}
\usepackage{todonotes}
\usepackage[capitalise]{cleveref}
\usepackage{microtype}
\definecolor{dkgreen}{rgb}{0.2,0.7,0.2}

\theoremstyle{definition}
\newtheorem{theo}{Theorem}[section]
\newtheorem{prop}[theo]{Proposition}
\newtheorem{lemma}[theo]{Lemma}
\newtheorem{conjecture}[theo]{Conjecture}
\newtheorem{cor}[theo]{Corollary}
\newtheorem{exa}[theo]{Example}
\newtheorem{defi}[theo]{Definition}
\newtheorem{rem}[theo]{Remark}
\theoremstyle{remark}

\numberwithin{equation}{section}

\newcommand{\N}{\mathbb{N}}

\newcommand{\F}{\mathbb{F}}

\newcommand{\cC}{\mathcal{C}}

\newcommand{\cG}{\mathcal{G}}

\newcommand{\cR}{\mathcal{R}}
\newcommand{\cS}{\mathcal{S}}
\newcommand{\cT}{\mathcal{T}}
\newcommand{\cU}{\mathcal{U}}
\newcommand{\cV}{\mathcal{V}}
\newcommand{\cW}{\mathcal{W}}

\newcommand{\PGnq}{\text{PG}(n-1,q)}
\newcommand{\Gkn}{\cG_q(k,n)}

\newcommand{\GL}{\text{GL}}
\newcommand{\PGL}{\text{PGL}}
\newcommand{\GLhat}{\text{GL}_{n/s}(q^s)}
\newcommand{\Fqn}{\F_{q^n}}
\newcommand{\Fqs}{\F_{q^s}}
\newcommand{\Fqa}{\F_{q^a}}
\newcommand{\cUhat}{\widehat{\cU}}
\newcommand{\cVhat}{\widehat{\cV}}
\newcommand{\mmid}{\!\mid\!}
\newcommand{\Gal}{\text{Gal}}
\newcommand{\Aut}{\text{Aut}}
\newcommand{\spann}{\text{span}}

\newcommand{\im}{\text{im}}
\newcommand{\subgroup}[1]{\mbox{$\langle{#1}\rangle$}}
\newcommand{\inner}[2]{\mbox{$\langle{#1}\!\mid\!{#2}\rangle$}}
\newcommand{\T}{^{\sf t}}
\newcommand{\ds}{\text{d}}
\newcommand{\dd}{\text{d}_{\rm{s}}}

\newcommand{\qbinom}[2]{\genfrac{[}{]}{0pt}{}{#1}{#2}}

\newcounter{alp}
\newcounter{ara}
\newcounter{rom}
\newenvironment{romanlist}{\begin{list}{(\roman{rom})\hfill}{\usecounter{rom}
			\topsep0ex \labelwidth.8cm \leftmargin.8cm \labelsep0cm
			\rightmargin0cm \parsep0ex \itemsep.4ex
			\partopsep1ex}}{\end{list}}
\newenvironment{alphalist}{\begin{list}{(\alph{alp})\hfill}{\usecounter{alp}
			\topsep0ex \labelwidth.7cm \leftmargin.7cm \labelsep0cm
			\rightmargin0cm \parsep0ex \itemsep0.4ex
			\partopsep0ex}}{\end{list}}

\DeclareMathOperator{\Orb}{Orb}

\begin{document}
		
	\title{Automorphism Groups and Isometries\\ for Cyclic Orbit Codes}
	\author{Heide Gluesing-Luerssen$^\ast$ and Hunter Lehmann\footnote{HGL was partially supported by the grant \#422479 from the Simons Foundation.
  HGL and HL are with the Department of Mathematics, University of Kentucky, Lexington KY 40506-0027, USA;
\{heide.gl, hunter.lehmann\}@uky.edu.}}

\date{January 23, 2021}
\maketitle
	
\begin{abstract}\label{sec:Abstract}
We study orbit codes in the field extension~$\Fqn$.
First we show that the automorphism group of a cyclic orbit code is contained in the normalizer of the Singer subgroup 
if the orbit is generated by a subspace that is not contained in a proper subfield of $\Fqn$. 
We then generalize to orbits under the normalizer of the Singer subgroup.
In that situation some exceptional cases arise and some open cases remain.
Finally we characterize linear isometries between such codes.
\end{abstract}
	
\section{Introduction}\label{sec:Introduction}

In~\cite{KoKsch08} Koetter/Kschischang introduced subspace codes for random network coding. 
As they demonstrated, these codes, together with rank-metric codes, are the appropriate tools for information transmission
with error correction through a network with multiple sources and receivers.
As a consequence,~\cite{KoKsch08} led to an intense study of both classes of codes.

Mathematically, a subspace code is simply a collection of subspaces of some vector space~$\F_q^n$, endowed with the subspace distance.
One class  that garnered particular attention are orbit codes; see~\cite{EKW10,KoKu08,TMBR13}.
These are, by definition, orbits of a subspace of~$\F_q^n$ under a subgroup of $\GL_n(\F_q)$ 
(acting naturally on the set of subspaces). 
If the group is cyclic, these codes are known as cyclic orbit codes.
However, in most of the literature the latter notion is reserved for orbit codes under the Singer subgroup, and we 
will follow this custom in this paper.

A Singer subgroup is, by definition, a cyclic subgroup of $\GL_n(\F_q)$ of order $q^n-1$.
Its meaning is best understood by identifying $\F_q^n$ with the field extension $\Fqn$ as $\F_q$-vector spaces. 
The matrix group $\GL_n(\F_q)$ turns into the group of $\F_q$-vector space automorphisms of $\Fqn$, and we will denote this group by $\GL_n(q)$.
The subgroup consisting of the multiplication maps $x\mapsto ax$ for any $a\in\Fqn^*$ is isomorphic to $\Fqn^*$ and 
thus a Singer subgroup of~$\GL_n(q)$.
In fact, all Singer subgroups of $\GL_n(q)$ are conjugate to $\Fqn^*$ and can be interpreted as a group of multiplication maps; see 
\cref{L-ConjSinger} and the paragraph thereafter.
In this setting, a cyclic orbit code is thus the orbit of an $\F_q$-subspace~$\cU$ of $\Fqn$ under $\Fqn^*$, i.e., $\{\omega^i\cU\mid i=0,\ldots,q^n-2\}$, where $\omega$ is a primitive element of~$\Fqn$.

First examples of cyclic orbit codes with good distance appeared already in \cite{EtVa11}, and in fact, in most of the literature, cyclic orbit codes have been studied in this setting, see for instance~\cite{GLMT15} for details on the orbit length and some distance results, \cite{BEGR16,OtOz17,RRT18,ChLi18,ZhTa19} for constructions of unions of cyclic orbit codes with good distance with the aid of subspace polynomials and \cite{GLL19} for a study of the distance distribution of cyclic orbit codes.
The aforementioned unions of cyclic orbit codes are in fact orbits codes under the normalizer of~$\Fqn^*$ in $\GL_n(q)$. 
The normalizer is isomorphic to $\Gal(\Fqn\mmid\F_q)\rtimes\Fqn^*$, and thus its orbits are simply unions of at most~$n$ 
cyclic orbit codes. 
This insight has also been utilized in \cite{BEOVW16}, where the authors succeeded in finding a $q$-Steiner system of type $\cS_2[2,3,13]$:
it consists of~$15$ orbit codes under the normalizer group.

In this paper we will study the automorphism groups of cyclic orbit codes and orbit codes under the Singer normalizer.
As usual, the automorphism group of a subspace code is defined as the subgroup of $\GL_n(q)$ that leaves the code invariant.
We will prove the following result.
Let~$\cU$ be a subspace of~$\Fqn$ containing~$1$ (which is no restriction) and let~$\Fqs$ be the smallest subfield of~$\Fqn$ containing~$\cU$.
Then the automorphism group is contained in the normalizer of the extension-field subgroup $\GLhat$, where the latter is defined as the 
subgroup of all $\Fqs$-linear automorphisms of~$\Fqn$. 
In particular, if~$\cU$ is generic, i.e., not contained in a proper subfield of~$\Fqn$, the automorphism group of the cyclic orbit code generated by~$\cU$ is contained in the normalizer of the Singer subgroup~$\Fqn^*$.
In order to prove these results we will derive a lower bound on the length of the $\GLhat$-orbit of~$\cU$ 
 for any given divisor~$s$ of~$n$.
A crucial role will be played by the parameter $\delta_s(\cU)$, which is the $\Fqs$-dimension of the $\Fqs$-subspace generated 
by~$\cU$.
Note that $\delta_s(\cU)=1$ iff $\cU\subseteq\Fqs$. 

We then turn to orbit codes under the normalizer of the Singer subgroup and derive the same results for the automorphism groups as long as the orbit code is generated by a subspace~$\cU$ satisfying $\delta_s(\cU)\neq2$. 
The case $\delta_s(\cU)=2$ is of particular interest:
the above results hold for many parameter cases, while there exist counterexamples for others.
We strongly believe that these examples are the only exceptions to our main result on the automorphism group.

We finally discuss linear isometries, i.e., maps from $\GL_n(q)$, between cyclic orbit codes and orbit codes under the Singer normalizer.
Our results on the automorphism groups immediately imply the following facts for orbits generated by generic subspaces:
(i) a linear isometry between cyclic orbit codes is in the normalizer of~$\Fqn^*$; 
(ii) linearly isometric orbit codes under the Singer normalizer are in fact equal -- with the possible exception of orbits generated by 
subspaces~$\cU$ with $\delta_s(\cU)=2$ for some~$s$.
This drastically reduces the work load for testing isometry between such codes.
The nature of our counterexamples leads us to believe that the last statement does not need the assumption on $\delta_s(\cU)$.
We close the paper with some examples listing the number of distinct isometry classes of cyclic orbit codes and, making use of~\cite{GLL19}, also provide the weight distribution for each class.

\section{Singer Subgroups and Extension-Field Subgroups}\label{S-Prelim}
Throughout we fix a finite field~$\F_q$. 
The field extension $\F_{q^n}$ is taken as our model for the $n$-dimensional vector space over~$\F_q$.
We denote by $\PGnq$ the \emph{$n$-dimensional projective geometry over} $\F_q$, that is, the set of all subspaces 
of $\Fqn$.
Accordingly, $\GL_n(q)$ denotes the group of all $\F_q$-automorphisms of~$\Fqn$.
Specific subgroups will play a crucial role.

\begin{defi}\label{D-GLhat}
Let $\Fqs$ be a subfield of~$\Fqn$, thus $\Fqn$ is an $\Fqs$-vector space of dimension $n/s$. 
The \emph{extension-field subgroup of degree~$s$} is defined as 
\[
    \GLhat=\{\phi\in\GL_n(q)\mid \phi\text{ is $\Fqs$-linear}\}.
\]
The subgroup~$\GL_1(q^n)$ will be identified with the multiplicative group~$\Fqn^*$ via the map $a\mapsto m_a$, where~$m_a$ is the multiplication by~$a$, that is, 
\begin{equation}\label{e-ma}
     m_a:\Fqn\longrightarrow\Fqn,\ x\longmapsto ax.
\end{equation}
\end{defi}

Clearly, $\GL_1(q^n)$ is a cyclic subgroup of order $q^n-1$.  Subgroups of $\GL_n(q)$ of this form are well known.

\begin{defi}\label{D-Singer}
A cyclic subgroup of $\GL_n(q)$ of order $q^n-1$ is called a \emph{Singer subgroup}.
\end{defi}

\begin{lemma}[\mbox{\cite[Lem.~3]{Gill16}}]\label{L-ConjSinger}
Every Singer subgroup of $\GL_n(q)$ is conjugate to $\Fqn^*$.
\end{lemma}

Let us briefly comment on this result.
Consider the extension-field subgroups $\GLhat$ from Definition~\ref{D-GLhat}, and let $\rho\in\GL_n(q)$.
Then the $\F_q$-linear isomorphism~$\rho$ leads to new field structures $\rho(\Fqn)$ and $\rho(\F_q)$ with identity $\rho(1)$ 
(they turn $\rho$ into a ring homomorphism).
The conjugate group $\rho\GLhat\rho^{-1}$ is now the group of all $\rho(\Fqs)$-linear automorphisms of the field $\rho(\Fqn)$, and in particular
the conjugate Singer subgroup $\rho\Fqn^*\rho^{-1}$ is the group of all $\rho(\Fqn)$-linear 
auto\-morphisms of the field $\rho(\Fqn)$.
Thus, conjugation of any of these subgroups corresponds to an isomorphic field structure.
For this reason we may and will restrict ourselves to the Singer subgroup $\Fqn^*$.

The following results will be needed later on and are well known.
The normalizer of a subgroup~$H$ in a group~$G$ is denoted by $N_G(H)$.

\begin{theo}\label{T-SingerNorm}
Let $S=\subgroup{\tau}\leq \GL_n(q)$ be a Singer subgroup. 
\begin{alphalist}
\item The normalizer of~$S$ is $N_{\GL_n(q)}(S)=\subgroup{\tau,\sigma}\cong\Gal(\Fqn\mmid\F_q)\rtimes S$,  
        where $\sigma\in\GL_n(q)$ is the Frobenius homomorphism of order~$n$.  
        Moreover, $N_{\GL_n(q)}(S)$ is self-normalizing in $\GL_n(q)$.
\item The only Singer subgroup contained in  $N_{\GL_n(q)}(S)$ is~$S$.
\item Let $H\leq \GL_n(q)$ such that $S\leq H$. Then there is a divisor~$s$ of~$n$ such that $\GLhat\unlhd H$.
\item $N_{\GL_n(q)}(\GLhat)\cong\Gal(\F_{q^s}\mmid\F_q)\rtimes\GLhat$.
\end{alphalist}
\end{theo}

\begin{proof}
(a) is in \cite[Ch.~II, Satz~7.3(a) and its proof]{Hup67}, (b) in \cite[Prop.~2.5]{CoRe04},
(c) is in \cite[p.~232]{Kan80} and \cite[Thm.~7]{Gill16}, and (d) is in \cite[Sec.~2]{Gill16}.
\end{proof}

The following is immediate.

\begin{cor}\label{C-NormMax}
Let $S\leq \GL_n(q)$ be a Singer subgroup. If~$n$ is an odd prime or $n=2$ and $q\geq3$, then $N_{\GL_n(q)}(S)$ is a maximal subgroup of $\GL_n(q)$.
\end{cor}

All of the above can, of course, be translated  into matrix groups.
In order to do so, we consider the following  isomorphism. 
Fix a primitive element $\omega$ of~$\Fqn$, and let $f=x^n-\sum_{i=0}^{n-1}f_ix^i\in\F_q[x]$ 
be its minimal polynomial over~$\F_q$.
Let 
\begin{equation}\label{e-CompMat}
   M_f=\begin{pmatrix} &1& &  \\ & &\ddots &  \\ & & & 1\\ f_0&f_1&\cdots&f_{n-1}\end{pmatrix}\in\F_q^{n\times n}
\end{equation}
be the companion matrix of~$f$. 
Then $1,\omega,\ldots,\omega^{n-1}$ form a basis of $\Fqn$ over~$\F_q$, and we 
have the isomorphism
\begin{equation}\label{e-FqnIso}
   \Phi:\Fqn\longrightarrow\F_q^n,\quad \sum_{i=0}^{n-1}a_i\omega^i\longmapsto (a_0,\ldots,a_{n-1}).
\end{equation}
It satisfies
\begin{equation}\label{e-omegaMat}
   \Phi(c\, \omega^i)=\Phi(c) M_f^i\ \text{ for all $c\in\Fqn$ and all }i\in\N_0.
\end{equation}
In other words, $M_f^i$ is the matrix representation of the linear map $m_{\omega^i}$ with respect to the basis $1,\omega,\ldots,\omega^{n-1}$.

\begin{rem}\label{R-GLhat}
Denote by $\GL_n(\F_q)$ the general linear group of invertible $n\times n$-matrices over~$\F_q$ and
identify a matrix $A\in\GL_n(\F_q)$ in the usual way with the isomorphism $\F_q^n\longrightarrow\F_q^n,\ v\longmapsto vA$.
Then we have the group isomorphism
\[
      \GL_n(\F_q)\longrightarrow \GL_n(q),\ A\longmapsto \phi_A=\Phi^{-1}\circ A\circ\Phi,
\]
which satisfies
\begin{equation}\label{e-PhiA}
    \phi_A(a)=\Phi^{-1}\big(\Phi(a)A\big) \text{ for all }a\in\Fqn.
\end{equation}
Let now $s$ be a divisor of~$n$ and set $N=(q^n-1)/(q^s-1)$. Thus~$\omega^N$ is a primitive element of $\Fqs$.
Then for any $A\in\GL_n(\F_q)$
\[
   \phi_A\text{ is $\Fqs$-linear}\Longleftrightarrow AM_f^N=M_f^NA.
\]
As a consequence, the subgroup $\{A\in\GL_n(\F_q)\mid AM_f^N=M_f^NA\}$ may be identified with the extension-field subgroup
$\GLhat$.
Consider the special case $s=1$.
From~\cite[Thm.~2.9]{Hes70} it is known that~$\subgroup{M_f}$ is self-centralizing, i.e., $\{A\in\GL_n(\F_q)\mid AM_f=M_fA\}=\subgroup{M_f}$
(see also \cite[Cor.~2 and Cor.~3]{GoCa93}). Since $\GL_1(q^n)\cong\F_{q^n}^*$, this simply reflects the well-known isomorphism
$\F_{q^n}^*\cong\subgroup{M_f}$ (and $\Fqn\cong\F_q[M_f]$).
\end{rem}

\section{Orbit Codes and Linear Isometries}\label{S-OrbCodes}
In this section we turn to subspace codes and, more specifically, orbit codes.
We endow the projective geometry $\PGnq$ with the \emph{subspace distance}
\begin{equation}\label{e-ds}
	\ds(\cV,\cW):=\dim\cV+\dim\cW-2\dim(\cV\cap\cW)
\end{equation}
for $\cV,\cW\in\PGnq$.
The subspace distance is a metric on $\PGnq$; see \cite[Lem.~1]{KoKsch08}.
A subset of $\PGnq$ with at least two elements is called a \emph{subspace code (of block length~$n$)}.
The \emph{subspace distance} of a code~$\cC$ is, as usual, 
\begin{equation}\label{e-distCC}
     \dd(\cC):=\min\{\ds(\cV,\cW)\mid\cV,\,\cW\in\cC,\,\cV\neq \cW\}.
\end{equation}

The subspace codes defined next are \emph{constant-dimension codes}, that is, they are contained in some $\Gkn$, where $\Gkn$
denotes the Grassmannian consisting of the $k$-dimensional subspaces of $\Fqn$.

\begin{defi}\label{GOrbit}
Let $G\leq\GL_n(q)$ be a subgroup and let $\cU\in\Gkn$. 
Then the $G$-orbit of $\cU$, defined as $\Orb_G(\cU)=\{\phi(\cU) \mid \phi\in G\}$, 
is called an \emph{orbit code}.
For a Singer subgroup~$S$, the orbit $\Orb_{S}(\cU)$ is called a \emph{cyclic orbit code}.
\end{defi}

Two classes of orbit codes will be in the focus of this paper: 
orbits under the Singer subgroup~$\Fqn^*$ and orbits under the normalizer of~$\Fqn^*$.
They take the following explicit form. 
Let~$\omega$ be a primitive element of~$\Fqn$. 
Furthermore, for $\cU\in\Gkn$ define $\cU^{[i]}:=\{u^{[i]}\mid u\in\cU\}$, where we use the standard notation $[i]:=q^i$.
Consider the Singer subgroup $\Fqn^*$ and its normalizer $N:=N_{\GL_n(q)}(\Fqn^*)\cong\Gal(\Fqn\mmid\F_q)\rtimes \Fqn^*$. Then
\begin{equation}\label{e-S&N-orbits}
    \Orb_{\Fqn^*}(\cU)=\{\omega^i\cU\mid i=0,\ldots,q^n-2\}\ \text{ and }
    \Orb_N(\cU)=\bigcup_{i=0}^{n-1}\Orb_{\Fqn^*}(\cU^{[i]}).
\end{equation}

For later reference we record the following simple fact about the sizes of these orbits.
\begin{rem}[\mbox{\cite[Cor.~3.13]{GLMT15}}]\label{R-SingerOrbSize}
Let $\cU\in\Gkn$. Suppose $\F_{q^t}$ is the largest subfield of~$\Fqn$ such that~$\cU$ is closed under multiplication by scalars from~$\F_{q^t}$ (i.e., $\cU$ is an $\F_{q^t}$-vector space with respect to the ordinary multiplication in $\Fqn$). Then 
\[
   |\Orb_{\Fqn^*}(\cU)|=\frac{q^n-1}{q^t-1}.
\]
As a consequence, $ |\Orb_{N}(\cU)|\leq n(q^n-1)/(q^t-1)$ for the normalizer $N:=N_{\GL_n(q)}(\Fqn^*)$.
\end{rem}

Let us return to general $G$-orbits. 
In matrix notation, they take the following form. This is the setting in which they have been studied in~\cite{TMBR13}.

\begin{rem}\label{R-OrbCodeMat}
Let $\cU\in\Gkn$ and $G\leq\GL_n(q)$. 
Define $\tilde{G}:=\{\Phi\circ \phi\circ\Phi^{-1}\mid \phi\in G\}$ and 
$\tilde{\cU}=\Phi(\cU)$, where~$\Phi$ is the isomorphism  from \eqref{e-FqnIso}.
Then $\tilde{G}\leq\GL_n(\F_q)$ and $\cU\subseteq\F_q^n$, and~\eqref{e-PhiA} shows that 
\[
           \Phi(\Orb_G(\cU))=\Orb_{\tilde{G}}(\tilde{\cU}):=\{\tilde{\cU} A\mid A\in\tilde{G}\}.
\]
\end{rem}

In this paper we want to study linear isometries between orbit codes. 

\begin{defi}\label{D-IsoPG}
An \emph{isometry} on $\PGnq$ is a distance-preserving map $\varphi: \PGnq \to \PGnq$, thus,
$\ds(\cU,\cV)=\ds(\varphi(\cU),\varphi(\cV))$ for all $\cU,\cV \in \PGnq$.
\end{defi}

It is clear that an isometry is bijective. 
In \cite[2.3--2.8]{Trau13} it has been shown that the dimension-preserving isometries are 
precisely the elements of the projective general semi-linear group $\GL_n(q)/Z \rtimes \Aut(\F_q)$,
where~$Z$  is the center of~$\GL_n(q)$, that is, 
$Z=\{m_a\mid a\in\F_q^*\}$ with~$m_a$ as in~\eqref{e-ma}.
Thanks to the  Fundamental Theorem of Projective Geometry, these are exactly the automorphisms 
(i.e., incidence-preserving bijections) of $\PGnq$.
In this paper we will only consider linear isometries, that is, maps in the projective linear group $\PGL_n(q)=\GL_n(q)/Z$.
Note that a map $\phi\in\GL_n(q)$ is in~$Z$ if and only if it fixes every $\F_q$-subspace of $\Fqn$, which is why we may factor out~$Z$.
For ease of notation, we will simply consider linear isometries in $\GL_n(q)$.
This will have no impact on our considerations (one can just factor out Z in all groups occurring below).

\begin{defi}\label{D-LinIso}
Let $G\leq\GL_n(q)$ and $\cU_1,\,\cU_2\in\Gkn$.
Consider the $G$-orbits $\cC_i=\Orb_G(\cU_i)$ for $i=1,2$.
Then $\cC_1$ and $\cC_2$ are called \emph{(linearly) isometric} if there exists an isomorphism $\psi\in\GL_n(q)$ such that
$\psi(\cC_1)=\cC_2$, where $\psi(\cC_1):=\{\psi(\cV)\mid \cV\in\cC_1\}$. 
In this case~$\psi$ is called a \emph{(linear) isometry} between~$\cC_1$ and $\cC_2$.
In the special case, where $G=S$ is a Singer subgroup and $\psi(\cC_1)=\cC_2$ for some $\psi\in N_{\GL_n(q)}(S)$,
we call the cyclic orbit codes $\Orb_S(\cU_1)$ and $\Orb_S(\cU_2)$ \emph{Frobenius-isometric} and~$\psi$ a \emph{Frobenius-isometry}.
\end{defi}

The terminology Frobenius-isometry is motivated by the fact that, thanks to \cref{T-SingerNorm}(a),
$N_{\GL_n(q)}(S)\cong\Gal(\Fqn\mmid\F_q)\rtimes S$.

Later in Section~\ref{S-Isom} we will see that -- just like for block codes with the Hamming metric -- not every weight-preserving 
bijection between cyclic orbit codes is an isometry. 
Hence not every such map extends to an isometry on $\PGnq$.

The following is easy to see. 

\begin{theo}[\mbox{see also \cite[Thm.~10]{TMBR13}}]\label{T-isoOrbCodes}
 Let $G \leq \GL_n(q),\;\psi\in \GL_n(q)$, and $\cU \in \Gkn$.
\begin{alphalist}
\item Set $G'=\psi G\psi^{-1}$ and 
         $\cU'=\psi(\cU)$. Then the orbit codes $\cC = \Orb_G(\cU)$ and $\cC'=\Orb_{G'}(\cU')$ are linearly isometric 
         with $\cC'=\psi(\cC)$.
\item Let $\cC = \Orb_G(\cU)$ and $\cC'=\psi(\cC)$. Then $\cC'=\Orb_{\psi G\psi^{-1}}(\cU')$ with $\cU'=\psi(\cU)$. 
        As a consequence, if $\psi\in N_{\GL_n}(G)$, then $\cC$ and $\cC'$ are isometric $G$-orbit codes.
\end{alphalist}
\end{theo}

In order to study isometries between cyclic orbit codes, we need to understand their automorphism groups. 
This is the subject of the next section. 
For these considerations it will suffice to restrict to orbit codes generated by subspaces $\cU\in\Gkn$, where $k\leq n/2$. 
In order to see this, we need to briefly introduce the dual code.
Let~$\omega$ be a primitive element of~$\Fqn$ and choose the symmetric, non-degenerate, $\F_q$-bilinear form 
$\inner{\cdot}{\cdot}$ on $\Fqn$ defined via $\inner{\omega^i}{\omega^j}=\delta_{i,j}$ for all $i,j=0,\ldots,n-1$ (this is simply the standard dot 
product on $\F_q^n$ under the isomorphism in~\eqref{e-FqnIso}).
Define the dual of a subspace~$\cW\leq\Fqn$ in
the usual way as $\cW^\perp=\{v\in\Fqn\mid \inner{v}{w}=0\text{ for all }w\in\cW\}$.
Clearly, $\dim\cW^\perp=n-\dim\cW$.
The \emph{dual} of a subspace code $\cC\subseteq\Fqn$ is simply defined as $\cC^\perp:=\{\cW^\perp\mid \cW\in\cC\}$.
We can now describe the dual of an orbit code.
For an $\F_q$-linear map~$\phi:\Fqn\longrightarrow\Fqn$ denote by~$\phi^\dagger$ its adjoint map, that is, the unique linear map satisfying
$\inner{\phi(x)}{y}=\inner{x}{\phi^\dagger(y)}$ for all $x,y\in\Fqn$.
Clearly $\phi^\dagger\in\GL_n(q)$ for any $\phi\in\GL_n(q)$.

\begin{rem}\label{R-DualOrbit}
Suppose $\cC=\Orb_G(\cU)$ for some subgroup~$G\leq\GL_n(q)$.
Then $\cC^\perp=\Orb_{G^\dagger}(\cU^\perp)$, where $G^\dagger=\{\phi^\dagger\mid \phi\in G\}$, which is clearly a subgroup of $\GL_n(q)$. 
This follows immediately from $\phi(\cU)^\perp=(\phi^\dagger)^{-1}(\cU^\perp)$.
We call $G^\dagger$ the \emph{adjoint group of} $G$.
\end{rem}

In the setting of \cref{R-OrbCodeMat}, where subgroups of the matrix group $\GL_n(\F_q)$ act on subspaces in $\F_q^n$, this fact
 also appears in \cite[Thm.~18]{TMBR13}.
 
The following surprising result tells us that the adjoint groups of all groups of interest in this paper are conjugate to the 
group itself, and even more, we may choose the same conjugation matrix for all these groups.

\begin{theo}\label{T-AdjGroups}
There exists a map $\rho\in\GL_n(q)$ such that 
\[
  \rho^{-1}G^\dagger\rho=G\ \text{ for all } G\in\{\Fqn^*,\Gal(\Fqn\mmid\F_q)\}\cup\{\GLhat\mid s\text{ divisor of }n\}.
\]
\end{theo}

The proof, which is not needed for the rest of this paper, is postponed to an appendix. 
Returning to our orbit codes, \cref{T-AdjGroups} along with \cref{R-DualOrbit} tells us that the dual of a $G$-orbit, where~$G$ is any of the 
groups above, is again an orbit of the same type, but with respect to an isomorphic field structure; see the paragraph following \cref{L-ConjSinger}.
The isomorphic field structure does not depend on the group.

All of this tells us that it suffices to study isometries (and automorphisms) for orbit codes generated by subspaces  of dimension 
at most~$n/2$. 
Hence from now on we only consider subspaces $\cU\in\Gkn$, where $k\leq n/2$.

\section{The Automorphism Groups of Singer Orbits}\label{S-AutGroupSinger}

In this section we will derive information about the automorphism groups of cyclic orbit codes. 
This will be sufficient to discuss isometries between cyclic orbit codes later in this paper.
In accordance with earlier notation we will consider automorphisms in $\GL_n(q)$ rather than $\PGL_n(q)=\GL_n(q)/Z$.

\begin{defi}\label{D-AutGroup}
Let $\cC\subseteq\text{PG}(n-1,q)$ be a subspace code.
The \emph{automorphism group} of~$\cC$ is defined as the group of linear isometries that fix~$\cC$, that is,
$\Aut(\cC):=\{\psi\in\GL_n(q)\mid \psi(\cC)=\cC\}$.
Any subgroup of $\Aut(\cC)$ is called \emph{a group of automorphisms} of~$\cC$.
\end{defi}

Clearly, for any $G\leq\GL_n(q)$ and any orbit code $\cC=\Orb_G(\cU)$, the group~$G$ is a group of automorphisms of~$\cC$.
Furthermore, if $H\leq\GL_n(q)$  then
\begin{equation}\label{e-GHOrb}
    H\leq \Aut(\Orb_G(\cU))\Longleftrightarrow \Orb_H(\cU)\subseteq\Orb_G(\cU).
\end{equation}
We will now focus on the case where~$\cC$ is a cyclic orbit code, that is, $\cC=\Orb_S(\cU)$ for some subspace~$\cU\leq\Fqn$ and a Singer subgroup~$S\leq\GL_n(q)$.
Thanks to \cref{L-ConjSinger} and \cref{T-isoOrbCodes}(a) it suffices to study the case where $S=\Fqn^*$.
The following result is immediate with \cref{T-SingerNorm}(c). 

\begin{prop}\label{P-GLhatAut}
Let $\cC=\Orb_{\Fqn^*}(\cU)$ be a cyclic orbit code. 
Then there exists a divisor~$s$ of~$n$ such that
$\GLhat\unlhd\Aut(\cC)\leq N_{\GL_n(q)}(\GLhat)$.
\end{prop}

The following notion will be convenient throughout.

\begin{defi}\label{D-GenericSubspace}
A subspace $\cU\subseteq\Fqn$ is called \emph{generic} if $\cU$ is not contained in a proper subfield of~$\Fqn$.
\end{defi}

The next theorem is the main result of this section. 
It shows that for any subspace~$\cU$,  the parameter~$s$ from \cref{P-GLhatAut} is
the smallest divisor of~$n$ such that $\cU\subseteq\Fqs$.
As a consequence, the automorphism group of $\cC$ contains linear isometries that are outside the normalizer of~$\Fqn^*$ if and only if
$\cU$ is not generic.

Since any cyclic orbit code contains a subspace~$\cU$ such that $1\in\cU$, we
may assume without loss of generality that~$1$ is contained in the generating subspace.
If, in addition, $\dim(\cU)=1$, then $\cU=\F_q$ and 
$\Orb_{\Fqn^*}(\cU)=\Orb_{\GLhat}(\cU)=\cG_q(1,n)$ for all divisors~$s$ of~$n$.
Hence from now on we assume $k\geq2$ and thus $n\geq4$.
  
\begin{theo}\label{T-AutOrb}
Let  $S=\Fqn^*$  and let $\cU\in\Gkn$ be such that $1\in\cU$. 
Let~$s$ be a divisor of~$n$. Then
\begin{equation}\label{e-UFqs}
    \cU\subseteq\Fqs\Longleftrightarrow\GLhat\leq\Aut(\Orb_S(\cU)).
\end{equation}
Moreover, if~$\Fqs$ is the smallest subfield containing~$ \cU$, then
$\GLhat$ is normal in $\Aut(\Orb_S(\cU))$ and thus 
$ \Aut(\Orb_S(\cU))\leq N_{\GL_n(q)}(\GLhat)$.
As a consequence, 
\[
   \cU\text{ is generic }\Longleftrightarrow \Aut(\Orb_S(\cU))\leq N_{\GL_n(q)}(S).
\]
\end{theo}

The proof is postponed to the end of this section. We first need some technical results.
We start with a lower bound on the size of the orbits $\Orb_{\GLhat}(\cU)$ for a given divisor~$s$ of~$n$.
As we will see, this size depends on the dimension of the $\Fqs$-subspace of $\Fqn$ generated by~$\cU$.

\begin{defi}\label{D-Uhat}
For any $\F_q$-subspace $\cV$ of $\Fqn$ we set 
$\cVhat:=\spann_{\Fqs}(\cV)$ and $\delta_s(\cV):=\dim_{\F_{q^s}}(\cVhat)$.
Note that $\delta_s(\cV)s=\dim_{\F_q}(\cVhat)\leq n$.
\end{defi}

Clearly $\delta_s(\,\cdot\,)$ is invariant under the actions of the groups~$\Fqn^*,\,\Gal(\Fqn\mmid\F_q)$, and $\GLhat$.

\begin{prop}\label{P-dimUhat}
Let $\cU\in\Gkn$ be such that $1\in\cU$, and let~$s$ be a divisor of~$n$.
Set $\delta_s(\cU)=r$. Then $1\leq r\leq k$ and 
\[
    |\Orb_{\GLhat}(\cU)| \geq  \frac{q^{\binom{r}{2}(s-1)}}{\qbinom{k}{r}_q} \prod_{i=0}^{r-1}\frac{q^{n-is}-1}{q^{r-i}-1}
\]
with equality if $r=k$. 
\end{prop}

Note that $(r-1)s<rs=\dim_{\F_q}(\cUhat)\leq n$. This shows that the right hand side is not~$0$.

\begin{proof}
First let $s=1$. Then $\Fqs=\F_q$ and $\cUhat=\cU$, and thus $r=k$. 
In this case, $\Orb_{\GL_n(q)}(\cU)$ consists of all $k$-dimensional subspaces of $\Fqn$, and hence its size is $\qbinom{n}{k}_q$, which is the right hand side above.
From now on let $s>1$. 
\\
\underline{Case 1)} Let $r=k$. 
Let $B=(u_1,\ldots,u_k)$ be an ordered $\F_q$-basis of $\cU$. 
Thanks to $\delta_s(\cU)=k$, the vectors $u_1,\ldots,u_k$ are also $\Fqs$-linearly independent. 
Under the action of $\GLhat$ the orbit of the basis~$B$ consists of all $k$-tuples of $\Fqs$-linearly independent vectors in $\Fqn$.
This implies that $|\Orb_{\GLhat}(\cU)|$ is given by the number of $k$-tuples of $\Fqs$-linearly independent vectors in $\Fqn$ divided by the number of ordered $\F_q$-bases for a $k$-dimensional $\F_q$-subspace.  
We conclude
\[
   |\Orb_{\GLhat}(\cU)| = \dfrac{\prod\limits_{i=0}^{k-1} (q^n-q^{is})}{\prod\limits_{i=0}^{k-1}(q^k-q^i)} 
	 					   = q^{\binom{k}{2}(s-1)} \prod_{i=0}^{k-1} \frac{q^{n-is}-1}{q^{k-i}-1}.
\]
\underline{Case 2)}  Let now $1\leq r<k$.
There exists a subspace $\cV$ of $\cU$ such that $\dim_{\F_q}(\cV)=r$ 
and $\widehat{\cV}=\cUhat$.  
Clearly, each subspace $\psi(\cU) \in \Orb_{\GLhat}(\cU)$ contains exactly $K:=\qbinom{k}{r}_q$ subspaces of $\F_q$-dimension~$r$, and thus in particular 
at most $K$ subspaces of the form~$\psi'(\cV)$ for some $\psi'\in\GLhat$.
Since each $\psi'(\cV) \in \Orb_{\GLhat}(\cV)$ is contained in at least one $\psi(\cU) \in \Orb_{\GLhat}(\cU)$, we obtain
\[ 
    |\Orb_{\GLhat}(\cU)| \geq \frac{1}{K}|\Orb_{\GLhat}(\cV)|. 
\]
Since~$\cV$ satisfies $\delta_1(\cV)=\delta_s(\cV)$,we may apply Case~1) to $\cV$ to obtain the desired result.
\end{proof}

In order to compare the sizes of the $\GLhat$-orbits and the Singer orbits, we need some technical lemmas.

\begin{lemma}\label{L-ExpDiffPos}
Let  $2\leq r\leq n/2$ and $1\leq s<n$ be such that $sr\leq n$.  Then  
\[
       \prod_{i=0}^{r-1}\frac{q^{n-is}-1}{q^{r-i}-1}>q^{r(n-r)-(s-1)\binom{r}{2}}.
\]
\end{lemma}

\begin{proof}
Note first that by the assumptions
\begin{equation}\label{e-Ineq1}
  n-is-r+i\geq1 \text{ for all }i=0,\ldots,r-1.
\end{equation}
Indeed, $n-is-r+i=n-r-i(s-1)\geq n-r-(r-1)(s-1)=n-rs+s-1$.
If $s\geq2$ the latter is clearly at least~$1$, while for $s=1$ we have $n-rs+s-1=n-r\geq n/2\geq1$, too.
Using the inequality 
\begin{equation}\label{e-qab}
      \frac{q^a-1}{q^b-1}> q^{a-b}\ \text{ whenever }a> b,
\end{equation}
we obtain from~\eqref{e-Ineq1} the inequality $\prod_{i=0}^{r-1}\frac{q^{n-is}-1}{q^{r-i}-1}>q^M$,
where $M=\sum_{i=0}^{r-1}(n-is-r+i)=r(n-r)-(s-1)\binom{r}{2}$, as desired.
\end{proof}

The next lemma comes in two forms, one with a factor~$n$ on the right hand side and one without such factor. 
The version with factor~$n$ will be needed in Section~\ref{S-AutGroupSingerNorm} when we study orbits under the normalizer of the 
Singer subgroup. 

\begin{lemma}\label{L-GLhatOrbitB}
Let $2\leq r\leq k \leq n/2$ and $1\leq s<n$ such that $rs\leq n$.  
\begin{alphalist}
\item Let $r\geq3$. Then 
         \begin{equation}\label{e-MainIneq}
                   q^{\binom{r}{2}(s-1)}\prod_{i=0}^{r-1}\frac{q^{n-is}-1}{q^{r-i}-1} >n\qbinom{k}{r}_q \frac{q^n-1}{q-1}.
         \end{equation}
\item If $r=2$, then ${\displaystyle q^{\binom{r}{2}(s-1)}\prod_{i=0}^{r-1}\frac{q^{n-is}-1}{q^{r-i}-1} >\qbinom{k}{r}_q \frac{q^n-1}{q-1}}$.
\end{alphalist}
\end{lemma}

\begin{proof}
(a) Let $r\geq3$, thus $n\geq6$. Setting $c=q/(q-1)$ we have $(q^n-1)/(q-1)<cq^{n-1}$. Furthermore, $r\geq3$ implies
\[
  r(k-r)+n-1\leq r(n-r)-(n/2+1),
\]
because $r(k-r)+n-1\leq r(n/2-r)+n-1=r(n-r)-r n/2+n-1\leq r(n-r)-3 n/2+n-1$.
Using the above inequalities along with $\qbinom{k}{r}_q < 4q^{r(k-r)}$ (see \cite[Lem.~4]{KoKsch08}) and  \cref{L-ExpDiffPos}
we compute
\[
  n\qbinom{k}{r}_q \frac{q^n-1}{q-1}<4ncq^{r(k-r)+n-1}< \frac{4nc}{q^{n/2+1}}q^{\binom{r}{2}(s-1)}\prod_{i=0}^{r-1}\frac{q^{n-is}-1}{q^{r-i}-1}.
\]
Finally, one easily checks that $\frac{4nc}{q^{n/2+1}}=\frac{4n}{q^{n/2}(q-1)}\leq 1$ for $q\geq 3$ and $n\geq4$ as well as $q=2$ and $n\geq 11$.
For the remaining cases ($q=2$ and $n=6,\ldots,10$) Inequality~\eqref{e-MainIneq} can be verified directly.
\\
(b)
Let $r=2$. In this case the desired inequality is equivalent to 
\[
       Q:=(q-1)(q^n-q^s)-(q^k-1)(q^k-q)>0.
\]
Since~$Q$ decreases with increasing~$s$ or~$k$, we may lower bound~$Q$ by using $s=k=n/2$ (ignoring that this may not be an integer).
This leads to
\begin{align*}
   Q&\geq(q-1)(q^n-q^{n/2})-(q^{n/2}-1)(q^{n/2}-q)
     =\big((q-1)q^{n/2}-(q^{n/2}-q)\big)(q^{n/2}-1)\\
    &=\big((q-2)q^{n/2}+q\big)(q^{n/2}-1)>0,
\end{align*}
as desired. 
\end{proof}

\begin{lemma}\label{L-GLhatSize}
Let $s,\,t\in\N$ such that $s\mmid t\mmid n$ and $s\neq t$. Then
\[
  \big|\GLhat\big|>\big|N_{\GL_n(q)}(\GL_{n/t}(q^t))\big|.
\]
\end{lemma}

\begin{proof}
Set $\hat{q}=q^s,\,\hat{n}=n/s$ and let $sa=t$. 
Then $\big|\GLhat\big|=\prod_{i=0}^{\hat{n}-1}(\hat{q}^{\hat{n}}-\hat{q}^i)$ and from \cref{T-SingerNorm}(d) we know that
\[
    \big|N_{\GL_n(q)}(\GL_{n/t}(q^t))\big|=t\prod_{i=0}^{\hat{n}/a-1}((\hat{q}^a)^{\hat{n}/a}-(\hat{q}^a)^i)
    \leq n\prod_{i=0}^{\hat{n}/a-1}(\hat{q}^{\hat{n}}-\hat{q}^{ai}).
\]
Clearly, all factors in the product on the right hand side appear in $|\GLhat|$.
Furthermore, since $a>1$, the factor $\hat{q}^{\hat{n}}-\hat{q}=q^n-q^s$ of $|\GLhat|$ does not appear in $|N_{\GL_n(q)}(\GL_{n/t}(q^t))|$.
Hence the desired inequality follows if we can show that $q^n-q^s>n$. 
Since $s\neq n$ and~$s$ is a divisor of~$n$, we have $q^n-q^s-n\geq q^n-q^{n/2}-n$. 
One easily verifies that the function $f(x)=q^x-q^{x/2}-x$ is indeed positive on $[4,\infty)$. 
This concludes the proof. 
\end{proof}

\bigskip

Now we are ready to prove our main result.

\noindent\emph{Proof of \cref{T-AutOrb}.}
Let~$S,\,s,\,\cU$ be as in the theorem.
Set $\cUhat=\spann_{\Fqs}(\cU)$ as in \cref{D-Uhat}.
Then $1\leq\delta_s(\cU)\leq\dim_{\F_q}(\cU)$ and
\begin{equation}\label{e-UUhat}
      \cU\subseteq\Fqs\Longleftrightarrow1=\delta_s(\cU)\Longleftrightarrow\cUhat=\Fqs,
\end{equation}
where the last equivalence follows from the fact that $1\in\cU$.
Since $|S|=q^n-1$ and $\F_q^*$ stabilizes~$\cU$, we have $|\Orb_S(\cU)|\leq(q^n-1)/(q-1)$ by the
orbit-stabilizer theorem (see also \cref{R-SingerOrbSize}).
Moreover, 
since $S\leq\GLhat$, we have 
\begin{equation}\label{e-OrbComp}
  \Orb_S(\cU)\subseteq\Orb_{\GLhat}(\cU) \ \text{ with equality iff }\ \GLhat\leq\Aut(\Orb_S(\cU)).
\end{equation}
We now prove the equivalence~\eqref{e-UFqs}. 
\\
``$\Longrightarrow$''
Let $\cU \subseteq \Fqs$, thus $\cUhat=\Fqs$.
Since $1\in\cU$  we obtain $\psi(\cU)=\{u\cdot \psi(1)\mid u\in\cU\}$ for every $\psi\in\GLhat$. 
Hence $\psi(\cU)$ is the cyclic shift $\psi(1)\cU$ and thus contained in  $\Orb_S(U)$. 
This shows $\Orb_{\GLhat}(\cU) \subseteq \Orb_S(\cU)$ and~\eqref{e-OrbComp} implies the desired result.
\\
``$\Longleftarrow$''
Suppose $\cU \not\subseteq\Fqs$. 
Then  $r:=\delta_s(\cU)\geq2$ by~\eqref{e-UUhat}.
\cref{P-dimUhat} and \cref{L-GLhatOrbitB} imply
\[ 
    \big|\Orb_{\GLhat}(\cU)\big| \geq\frac{ q^{\binom{r}{2}(s-1)}}{\qbinom{k}{r}_q}\prod_{i=0}^{r-1} \frac{q^{n-is}-1}{q^{r-i}-1}
    >\frac{q^n-1}{q-1}\geq|\Orb_S(\cU)|.
\]
\eqref{e-OrbComp} implies $\GLhat\not\leq\Aut(\Orb_S(\cU))$. 

We now turn to the remaining statements of \cref{T-AutOrb}.
Let $\Fqs$ be the smallest subfield containing~$\cU$. We want to show that $\GLhat$ is normal in $\Aut(\Orb_S(\cU))$.
To this end set $\cT=\big\{t\in\N\,\big|\, s\mmid t\mmid n\big\}$.
Clearly $\GL_{n/t}(q^t)\leq\GLhat$ for all $t\in\cT$.
From~\eqref{e-UFqs} we conclude that for any $t\in\N$ 
\[
    \GL_{n/t}(q^t)\leq\Aut(\Orb_S(\cU))\Longleftrightarrow t\in\cT.
\]
Furthermore, thanks to \cref{T-SingerNorm}(c) one of the subgroups $\GL_{n/t}(q^t),\,t\in\cT$, is normal in $\Aut(\Orb_S(\cU))$.
Suppose $\GL_{n/t}(q^t)$ is normal in $\Aut(\Orb_S(\cU))$ for some $t\in\cT\setminus\{s\}$.
Then $\Aut(\Orb_S(\cU))\leq N_{\GL_n(q)}(\GL_{n/t}(q^t))$.
Now, \cref{L-GLhatSize} along with $\GLhat\leq\Aut(\Orb_S(\cU))$ leads to a contradiction. 
Thus $\GLhat$ is the only extension-field subgroup that is normal in $\Aut(\Orb_S(\cU))$.
The rest of the theorem follows. 
\hfill$\square$

\section{The Automorphism Groups of Orbits under the Singer Normalizer}\label{S-AutGroupSingerNorm}

The considerations of the previous sections allow us to also describe the automorphism group of orbits under the 
normalizer of the Singer subgroup in most cases.
Recall the notation in~\eqref{e-S&N-orbits}.
The following theorem is analogous to \cref{T-AutOrb}, but needs the assumption $\delta_s(\cU)\neq2$. 
We will deal with the case $\delta_s(\cU)=2$ afterwards.

Throughout this section, let $N:=N_{\GL_n(q)}(\Fqn^*)$, i.~e.,~$N$ is the normalizer of~$\Fqn^*$.
Recall also that we assume  $2\leq k\leq n/2$.

\begin{theo}\label{C-AutOrbNormalizer}
Let~$s$ be a divisor of~$n$ and $\cU\in\Gkn$ be such that $1\in\cU$ and such that  $\delta_s(\cU)\neq2$.
Then
\begin{equation}\label{e-UFqsNorm}
    \cU\subseteq\Fqs\Longleftrightarrow\GLhat\leq\Aut(\Orb_N(\cU)).
\end{equation}
Moreover, if~$\Fqs$ is the smallest subfield containing~$ \cU$ and $\delta_t(\cU)\neq2$ for all divisors~$t$ of~$n$, then
$\GLhat$ is normal in $\Aut(\Orb_N(\cU))$ and thus 
$ \Aut(\Orb_N(\cU))\leq N_{\GL_n(q)}(\GLhat)$. 
As a consequence:
\begin{alphalist}
\item If $\cU$ is generic and $\delta_t(\cU)\neq2$ for all divisors~$t$ of~$n$, then $\Aut(\Orb_N(\cU))= N$;
\item If $\Aut(\Orb_N(\cU))= N$, then~$\cU$ is generic.
\end{alphalist}
\end{theo}

Note that the left hand side of~\eqref{e-UFqsNorm} means that $\delta_s(\cU)=1$. 
Hence the excluded case $\delta_s(\cU)=2$ may be regarded as a transitional case, and we will see below that in that case
either is possible: $\GLhat\leq\Aut(\Orb_N(\cU))$ or $\GLhat\not\leq\Aut(\Orb_N(\cU))$.

\begin{proof}
For ``$\Longrightarrow$'' of \eqref{e-UFqsNorm} recall that $\Orb_N(\cU)=\bigcup_{i=0}^{n-1}\Orb_{\Fqn^*}(\cU^{[i]})$, 
see~\eqref{e-S&N-orbits}.
Since $1\in\cU^{[i]}$ and $\cU^{[i]}\subseteq\Fqs$ for all~$i$, the desired statement follows from \cref{T-AutOrb}.
\\
``$\Longleftarrow$'' The proof is similar to the one of \cref{T-AutOrb}.
Suppose $\cU \not\subseteq\Fqs$. 
Thanks to our assumption this implies  $\delta_s(\cU)=:r\geq3$.
Thus \cref{P-dimUhat}, \cref{L-GLhatOrbitB}(a), and \cref{R-SingerOrbSize} lead to
\[ 
    \big|\Orb_{\GLhat}(\cU)\big| \geq\frac{ q^{\binom{r}{2}(s-1)}}{\qbinom{k}{r}_q}\prod_{i=0}^{r-1} \frac{q^{n-is}-1}{q^{r-i}-1}
    >n\frac{q^n-1}{q-1}\geq|\Orb_N(\cU)|,
\]
and therefore $\GLhat\not\leq\Aut(\Orb_N(\cU))$. 
The rest of the proof is identical to the one for \cref{T-AutOrb}. 
For Part~(b) notice that ``$\Rightarrow$'' of \eqref{e-UFqsNorm} holds true for general $\delta_s(\cU)$.
\end{proof}

We now turn to the remaining case  $r=\delta_s(\cU)=2$. 
In this case there are indeed instances where $\GLhat\leq\Aut(\Orb_N(\cU))$ even though $\cU\not\subseteq\Fqs$.
Clearly, this containment is equivalent to $\Orb_{\GLhat}(\cU)\subseteq\Orb_N(\cU)$.
In all known examples we even have $\Orb_{\GLhat}(\cU)=\Orb_N(\cU)$.
In fact, we believe that we have $\Orb_N(\cU)\subseteq\Orb_{\GLhat}(\cU)$ for all subspaces~$\cU$ 
(i.e., $\cU^q=\phi(\cU)$ for some $\phi\in\GLhat$), but unfortunately we are not able at this point to prove this statement.

\begin{exa}\label{E-SmallNormSingerGLhatOrb}
Let $(q,n,k,s)=(2,4,2,2)$.
A $2$-dimensional subspace $\cU\leq\F_{2^4}$ with $\delta_2(\cU)=2$ is of the form $\cU=\spann_{\F_2}\{1,\alpha\}$ for some 
$\alpha\in\F_{2^4}\setminus\F_{2^2}$.
One can directly verify (using, e.g., SageMath) that all these subspaces generate the same $N$-orbit, and this orbit agrees with 
the $\GL_2(4)$-orbit.
The orbit size is $n/2(2^n-1)/(2-1)=30$ (see also \cref{P-UFaFa} below). 
\end{exa}

\begin{exa}\label{E-NormSingerGLhatOrb}
Let $(q,n,k,s)=(2,8,4,4)$.
Let $\alpha\in\F_{2^8}\setminus\F_{2^4}$ and consider the subspace $\cU=\spann_{\F_{2^2}}\{1,\alpha\}$.
Then $\cU\not\subseteq\F_{2^4}$, hence $\delta_4(\cU)=2$, and one straightforwardly verifies that
$\Orb_{\GL_2(2^4)}(\cU)=\Orb_N(\cU)$, and the orbit has size $340$ (for comparison, the lower bound from \cref{P-dimUhat} is 292).
These observations can also be seen as follows. Let $S=\F_{2^8}^*$.
\begin{romanlist}
\item By \cref{R-SingerOrbSize} the Singer orbit has size $|\Orb_{S}(\cU)|=(2^n-1)/(2^2-1)=85$.
\item As \cref{P-UFaFa} below shows, $\cU^{[4]}\in\Orb_S(\cU)$; thus $\sigma^4$ stabilizes $\Orb_S(\cU)$,  where~$\sigma$ is the 
       Frobenius automorphism.
        Furthermore, no other non-trivial element of the Galois group $\Gal(\F_{2^8}\mmid \F_2)$ stabilizes $\Orb_S(\cU)$ (this is true for these specific parameters, but not in the general situation of \cref{P-UFaFa}).
         Together with~(i) this shows that $|\Orb_N(\cU)|= 4\!\cdot\!85=340$.
\item Since $\cU=\spann_{\F_{2^2}}\{1,\alpha\}$ and $\F_{2^2}\subseteq\F_{2^4}$, an $\F_{2^4}$-linear isomorphism~$\phi$ maps~$\cU$ 
        to the space $\spann_{\F_{2^2}}\{\phi(1),\phi(\alpha)\}$.         
        As a consequence, $\Orb_{\GL_2(2^4)}(\cU)$ consists of all subspaces in $\F_{2^8}$ that are $2$-dimensional over~$\F_{2^2}$ and 
        not $1$-dimensional over~$\F_{2^4}$, i.e., not a cyclic shift of $\F_{2^4}$. 
        Thus $|\Orb_{\GL_2(2^4)}(\cU)|=\qbinom{4}{2}_{4}-|\Orb_S(\F_{2^4})|=\qbinom{4}{2}_{4}-(2^8-1)/(2^4-1)=340$.
\item Finally, $\Orb_N(\cU)\subseteq \Orb_{\GL_2(2^4)}(\cU)$. 
        To see this, it suffices to show that  $\cU^{[i]}\in\Orb_{\GL_2(2^4)}(\cU)$ for all $i\in\{0,\ldots,n-1\}$.
       Note that $\cU^{[i]}=\spann_{\F_{2^2}}\{1,\alpha^{[i]}\}$.
       Since~$1$ and $\alpha^{[i]}$ are $\F_{2^4}$-linearly independent, there exists  $\phi\in\GL_2(2^4)$ such that $\phi(1)=1$ and 
       $\phi(\alpha)=\alpha^{[i]}$. Hence $\cU^{[i]}=\phi(\cU)\in\Orb_{\GL_2(2^4)}(\cU)$.
\end{romanlist}
We wish to add that all subspaces of the form $\spann_{\F_{2^2}}\{1,\alpha\}$ with $\alpha\in\F_{2^8}\setminus\F_{2^4}$ generate the same 
orbit, and this is the only $N$-orbit of a $4$-dimensional subspace that coincides with the $\GLhat$-orbit.
Finally, since~$\cU$ is actually an $\F_4$-vector space and $\F_{2^8}=\F_{4^4}$, we may regard all of this also as an 
example for the parameters $(q,n,k,s)=(4,4,2,2)$.
Thus $\Orb_{\GL_2(4^2)}(\cU)=\Orb_{N'}(\cU)$, where $N'=N_{\GL_4(4)}(\F_{4^4}^*)$.
\end{exa}

The subspaces~$\cU$ in the above examples are both of the form $\cU=\spann_{\Fqa}\{1,\alpha\}\subseteq\F_{q^n}$,
where $a=n/4$ and $s=k=n/2$ and $\cU\not\subseteq\Fqs$.
In \cref{C-UFaFa} below we will show that for no other subspaces of this type the $\GLhat$-orbit coincides with the $N$-orbit.
We start with showing that  all such subspaces~$\cU$  satisfy $\cU^{[s]}\in\Orb_{\Fqn^*}(\cU)$. 

\begin{prop}\label{P-UFaFa}
Let $a\in\N,\,n=4a$, and $s=k=2a$. Choose $\alpha\in\Fqn\setminus\Fqs$ and set
$\cU=\spann_{\Fqa}\{1,\alpha\}\subseteq\F_{q^n}$.
Then $\cU^{[s]}\in\Orb_{\Fqn^*}(\cU)$. Thus 
\[
        |\Orb_N(\cU)|\leq \frac{n}{2}\,\frac{q^n-1}{q^a-1}.
\]
\end{prop}

\begin{proof}
By \cref{R-SingerOrbSize} we have $|\Orb_{\Fqn^*}(\cU)|=(q^n-1)/(q^a-1)$. Thus the second statement follows once we establish 
$\cU^{[s]}\in\Orb_{\Fqn^*}(\cU)$.
To do so we proceed as follows.
\\
1) We show first that
\begin{equation}\label{e-alphasU}
     \alpha^{[s]}\cU\cap\cU\neq\{0\}.
\end{equation}
Since both $\cU$ and $\alpha^{[s]}\cU=\spann_{\Fqa}\{\alpha^{[s]},\alpha\alpha^{[s]}\}$ have dimension 
$2a=n/2$, \eqref{e-alphasU} is equivalent to $\alpha^{[s]}\cU+\cU\neq\Fqn$. Hence we have to show that 
$1,\,\alpha,\,\alpha^{[s]},\,\alpha\alpha^{[s]}$ are linearly dependent over $\Fqa$.
We show that there exist $\lambda,\mu,\nu\in\Fqa$ such that 
\begin{equation}\label{e-lincombo}
     \lambda+\mu\alpha+\mu\alpha^{[s]}+\nu\alpha\alpha^{[s]}=0.
\end{equation}
Raising~\eqref{e-lincombo} to the power $[a]$ and using $s=2a$ we obtain a second equation, which together with~\eqref{e-lincombo}
can be written as 
\begin{equation}\label{e-MatEq}
    \begin{pmatrix}1&\alpha+\alpha^{[2a]}& \alpha\alpha^{[2a]}\\ 1&\alpha^{[a]}+\alpha^{[3a]}& \alpha^{[a]}\alpha^{[3a]}\end{pmatrix}
    \begin{pmatrix}\lambda\\ \mu\\ \nu\end{pmatrix}=0.
\end{equation}
The matrix is row equivalent to
\[
    \begin{pmatrix}1&\alpha+\alpha^{[2a]}& \alpha\alpha^{[2a]}\\ 0&\alpha^{[a]}+\alpha^{[3a]}-\alpha-\alpha^{[2a]}& 
               \alpha^{[a]}\alpha^{[3a]}-\alpha\alpha^{[2a]}\end{pmatrix}.
\]    
Now we can find a solution of the desired form.
Suppose first that $\alpha-\alpha^{[a]}+\alpha^{[2a]}-\alpha^{[3a]}\neq0$. 
Set~$\nu=1$.
Then~\eqref{e-MatEq} has the unique (normalized) solution
\[
   \nu=1,\quad \mu=\frac{\alpha^{[a]}\alpha^{[3a]}-\alpha\alpha^{[2a]}}{\alpha-\alpha^{[a]}+\alpha^{[2a]}-\alpha^{[3a]}},\quad 
   \lambda=-\mu(\alpha+\alpha^{[2a]})-\alpha\alpha^{[2a]}.
\]
Using that $4a=n$, one easily verifies that $\mu^{[a]}=\mu$ and $\lambda=\lambda^{[a]}$, and thus $(\lambda,\mu,\nu)\in\Fqa^3$.
If \mbox{$\alpha\!-\!\alpha^{[a]}\!+\!\alpha^{[2a]}\!-\!\alpha^{[3a]}=0$}, \eqref{e-lincombo} has the solution $(\lambda,\mu,\nu)=(-(\alpha+\alpha^{[2a]}),1,0)$, which again is in $\Fqa^3$.
All of this establishes~\eqref{e-alphasU}.
\\
2) \eqref{e-alphasU} implies that also $\alpha^{-[s]}\cU\cap\cU\neq\{0\}$. 
Choose $\delta\in\alpha^{-[s]}\cU\cap\cU\setminus\{0\}$ and 
let $\gamma\in\Fqn^*$ be such that $\gamma^{[s]}=\delta$. 
Then $\gamma=\gamma^{[2s]}=\delta^{[s]}\in\cU^{[s]}$. Moreover, $ \gamma^{[s]}\alpha^{[s]}\in\cU$ and thus $\gamma\alpha\in\cU^{[s]}$.
All of this shows that $\gamma\cU=\spann_{\Fqa}\{\gamma,\gamma\alpha\}=\cU^{[s]}$.
Thus, $\cU^{[s]}\in\Orb_{\Fqn^*}(\cU)$, as desired. 
\end{proof}

\begin{rem}\label{R-UFaFa}
\cref{P-UFaFa} only provides an upper bound for $|\Orb_N(\cU)|$.
In fact, there even exist subspaces~$\cU$ of the specified form for which $\cU^{[i]}\in \Orb_S(\cU)$ for all~$i$ and thus 
$\Orb_N(\cU)=\Orb_S(\cU)$; for instance for $q=3$ and $a=2$. 
On the other hand, for $q=2$ and $a=2$ we have equality in \cref{P-UFaFa} for all subspaces of the given form.
\end{rem}

\begin{cor}\label{C-UFaFa}
Let the data be as in \cref{P-UFaFa}. Then 
\[
        \Orb_{\GLhat}(\cU)=\Orb_N(\cU) \Longleftrightarrow (q,a)\in\{(2,1),\,(2,2)\}.
\]
\end{cor}

\begin{proof}
``$\Longleftarrow$'' Examples~\ref{E-SmallNormSingerGLhatOrb} and~\ref{E-NormSingerGLhatOrb}.
\\
``$\Longrightarrow$'' Let $(q,a)\not\in\{(2,1),\,(2,2)\}$. We show that $|\Orb_{\GLhat}(\cU)|>|\Orb_N(\cU)|$. 
Thanks to \cref{P-UFaFa} it suffices to show $|\Orb_{\GLhat}(\cU)|-n/2(q^n-1)/(q^a-1)>0$.
Using $r=\delta_s(\cU)=2$ and $k=2a=s=n/2$ along with the lower bound in  \cref{P-dimUhat} the inequality follows if we prove
$Q:=q^{2a-1}(q^a-1)-2a(q^{2a-1}-1)>0$.
We have
\[
      Q>(q^{2a-1}-1)(q^a-1-2a).
\]
The first factor is clearly positive.
As for the second factor, note that the function $f(x)=q^x-(2x+1)$ is non-negative on $[1,\infty)$ if $q\geq3$, while for $q=2$ this is the case 
for the interval $[3,\infty)$.
This shows that $Q>0$ whenever $(q,a)\not\in\{(2,1),\,(2,2)\}$ and concludes the proof.
\end{proof}

There is one more known example where the $\GLhat$-orbit coincides with the $N$-orbit even though the subspace is 
not contained in $\Fqs$. In fact, it is the only such  example for $s=1$.
Indeed, note that $s=1$ together with $\delta_s(\cU)=2$ forces $\dim(\cU)=2$.
In \cref{E-k2} below we will list all $2$-dimensional subspaces for which the orbits coincide.

\begin{exa}\label{E-qns=251}
Let $(q,n,s)=(2,5,1)$ and choose any subspace $\cU\in\cG_2(2,5)$. Then $\Orb_{\GL_5(2)}(\cU)$ is the entire Grassmannian $\cG_2(2,5)$.
It has cardinality $\qbinom{5}{2}_2=155=5(2^5-1)/(2-1)$ and satisfies
$\Orb_{\GL_{5}(2)}(\cU)=\Orb_N(\cU)$. 
\end{exa}

We now turn to cases where Inequality~\eqref{e-MainIneq} of \cref{L-GLhatOrbitB}(a) holds true even with $r=2$.
Recall that $\delta_s(\cU)=2$ implies $s\leq n/2$ because $\delta_s(\cU)s\leq n$.

\begin{prop}\label{L-NewIneq}
Let $k\leq 3n/8$ and $s\leq n/2$ be a divisor of~$n$. Let $\cU\in\Gkn$ be such that $1\in \cU$.
Then
\begin{alphalist}
\item If $\delta_s(\cU)=2$, then $|\Orb_{\GLhat}(\cU)\big|> |\Orb_N(\cU)|$ and thus $\GLhat\not\leq\Aut(\Orb_N(\cU))$.
\item If~$\Fqs$ is the smallest subfield containing~$ \cU$, then $\GLhat$ is normal in $\Aut(\Orb_N(\cU))$ and thus 
         $ \Aut(\Orb_N(\cU))\leq N_{\GL_n(q)}(\GLhat)$. 
\item $\cU$ is generic iff $\Aut(\Orb_N(\cU))= N$.
\end{alphalist}
\end{prop}

\begin{proof}
(a) Let $r:=\delta_s(\cU)=2$.
We show that~\eqref{e-MainIneq} holds true for most parameters 
and discuss the remaining values subsequently.
Inequality~\eqref{e-MainIneq} is equivalent to
\begin{equation}\label{e-Q}
       Q:=(q-1)(q^n-q^s)-n(q^k-1)(q^k-q)>0.
\end{equation}
The left hand side decreases for increasing~$s$, and thus we may assume $s=n/2$. 
With the aid of~\eqref{e-qab} we compute
\begin{align}
     Q&\geq(q-1)q^{n/2}(q^{n/2}-1)-nq(q^k-1)(q^{k-1}-1)\nonumber \\
        &>\big((q-1)q^{n/2}q^{n/2-k+1}-nq(q^k-1)\big)(q^{k-1}-1) \nonumber \\
        &=\Big(\frac{(q-1)q^{n-k}}{q^k-1}-n\Big)q(q^k-1)(q^{k-1}-1)\nonumber \\
        &>\big((q-1)q^{n-2k}-n\big)q(q^k-1)(q^{k-1}-1)\label{e-Qr=2}\\
        &>\big((q-1)q^{n/4}-n\big)q(q^k-1)(q^{k-1}-1)\nonumber,
\end{align}
where in the last step we used that $k\leq 3n/8$.
Clearly the last three factors are positive.
As for the first factor, consider the function $f(x)=(x-1)x^{n/4}-n$. 
For fixed~$n$ the function is increasing on $[2,\infty)$. Furthermore, 
\[
   f(2)\geq 0\text{ for }n\geq 16,\quad f(3)\geq 0\text{ for }n\geq 8,\ f(4)\geq 0\text{ for }n\geq 4.
\]
Thus $Q>0$ if (i) $q\geq4$ and $n\geq4$, (ii) $q=3$ and $n\geq8$, or (iii) $q=2$ and $n\geq16$.
For the cases $q=2$ with $4\leq n\leq 15$ and $q=3$ with $4\leq n\leq 7$, 
direct verification shows that~\eqref{e-Q} holds true unless $(q,n,k)\in\{(2,8,3),(2,11,4)\}$.
We consider these cases separately.
\\
i) Let $(q,n,k)=(2,11,4)$. Then $s=1$ (since~$s$ is a divisor of~$n$). 
But then every $4$-dimensional subspace~$\cU$ satisfies $\delta_s(\cU)=4$, and thus there is nothing to show.
\\
ii) Let  $(q,n,k)=(2,8,3)$. In this case $s\in\{2,4\}$. 
Exhaustive consideration of all $3$-dimensional subspaces~$\cU$ in $\F_{2^8}$ with $\delta_s(\cU)=2$ shows that in each case
the orbit $\Orb_{\GLhat}(\cU)$ is strictly larger than $n(2^n-1)=2040$, which is an upper bound for $|\Orb_N(\cU)|$. 
To be precise, for $s=2$, there is exactly one $\GLhat$-orbit and it has size $5355$, while for $s=4$ there exists one 
orbit of size $61200$, two orbits of size $15300$, and one orbit of size $5100$. 
For comparison, the lower bound in \cref{P-dimUhat} only provides $|\Orb_{\GLhat}(\cU)|\geq1530$ if $s=2$ and 
$|\Orb_{\GLhat}(\cU)|\geq1458$ if $s=4$.
\\
For (b) and~(c) note that Part~(a) and \cref{C-AutOrbNormalizer} imply the equivalence $[\cU\subseteq\F_{q^t}\Longleftrightarrow\GL_{n/t}(q^t)\leq\Aut_N(\cU)]$
for any divisor~$t$ of $n$ with $t\leq n/2$.
Thus the proof follows as in \cref{T-AutOrb}.
\end{proof}

Now we can fully cover the case where $k=r=2$. Let $N:=N_{\GL_n(q)}(\Fqn^*)$. 

\begin{prop}\label{E-k2}
Let $n\geq4$ and $1\leq s\leq n/2$ be a divisor of~$n$.
The following are equivalent.
\begin{romanlist}
\item There exists $\cU\in\cG_q(2,n)$ such that $\delta_s(\cU)=2$ and $\Orb_{\GLhat}(\cU)=\Orb_N(\cU)$.
\item $(q,n,s)\in\{(2,4,2),(2,5,1),(4,4,2)\}$.
\end{romanlist}
\end{prop}

\begin{proof}
``(ii)~$\Rightarrow$~(i)'' Examples~\ref{E-SmallNormSingerGLhatOrb}, \ref{E-qns=251}, and \ref{E-NormSingerGLhatOrb}.
 \\[.7ex]
``(i)~$\Rightarrow$~(ii)''  By \cref{L-NewIneq} we must have $k=2>3n/8$, hence $n\leq 5$.
Since~$s$ is a divisor of~$n$ and $s\leq n/2$, this leaves the cases $(n,s)\in\{(4,1),(4,2),(5,1)\}$ with arbitrary~$q$.
Using \cref{P-dimUhat} for the case $r=k=2$ and $|\Orb_N(\cU)|\leq n(q^n-1)/(q-1)$, we conclude that 
$|\Orb_{\GLhat}(\cU)\big|>|\Orb_N(\cU)|$ if $Q:=q^{s-1}(q^{n-s}-1)-n(q^2-1)>0$.
\\
\underline{Case 1:} $(n,s)=(4,1)$.
\\
In this case $Q>0$ iff $q\geq4$. Thus it remains to consider $q\in\{2,3\}$.
Since $s=1$, every $2$-dimensional subspace~$\cU$ satisfies $\delta_s(\cU)=2$ and $|\Orb_{\GL_{4}(q)}(\cU)|=\qbinom{4}{2}_q$.
Furthermore, exhaustive verification shows that $|\Orb_N(\cU)|\leq n/2(q^n-1)/(q-1)$. Thus
$|\Orb_{\GL_{4}(q)}(\cU)|>|\Orb_N(\cU)|$.
\\
\underline{Case 2:} $(n,s)=(4,2)$. 
\\
In this case $Q>0$ for all $q\geq5$, and exhaustive verification shows that for $q=3$ every $2$-dimensional subspace~$\cU$ in 
$\F_{3^4}$ with $\delta_2(\cU)=2$ satisfies $|\Orb_N(\cU)|\leq n/2(q^n-1)/(q-1)<|\Orb_{\GL_{4}(3)}(\cU)|$ (where the first inequality also follows from \cref{P-UFaFa}). This leaves the  cases $(q,n,s)\in\{(2,4,2),(4,4,2)\}$.
\\
\underline{Case 3:} $(n,s)=(5,1)$. In this case $Q>0$ iff $q\geq3$, and thus only $(q,n,s)=(2,5,1)$ remains.
\end{proof}

Similarly we can cover all cases where $k=3$ (hence $n\geq6$).
In this case, $\Orb_{\GLhat}(\cU)$ is always strictly bigger than $\Orb_N(\cU)$.

\begin{prop}\label{E-k3}
Let $n\geq6$ and $s\leq n/2$ be a divisor of~$n$. 
Then for every subspace $\cU\in\cG_q(3,n)$ such that  $\delta_s(\cU)=2$ we have $|\Orb_{\GLhat}(\cU)|>|\Orb_N(\cU)|$.
\end{prop}

\begin{proof}
By Lemma~\ref{L-NewIneq} we only need to verify the cases where $k=3>3n/8$, thus $n<8$.
Since $\delta_s(\cU)=2\neq\dim(\cU)=3$, we must have $s\neq1$.
This leaves $(n,s)\in\{(6,2),(6,3)\}$. 
Using $n=6,\,k=3$ and $s\in\{2,3\}$ one verifies that~\eqref{e-Q} is true whenever $q\geq7$. 
Exhaustive verification for $q\in\{2,3,4,5\}$ establishes the desired result.
\end{proof}

As the proofs in this section have shown, for given parameters $(n,k,s)$ and $r=2$ the inequality in~\eqref{e-MainIneq} is true 
for sufficiently large~$q$ (for instance, if $k<n/2$, then this is the case for $q\geq n+1$ as~\eqref{e-Qr=2} shows).
Thus, any further examples where the $N$-orbit agrees with the $\GLhat$-orbit requires a relatively small field size.
We strongly believe that no further example exists and thus close this section with 

\begin{conjecture}\label{C-Conjec}
Let $s\leq n/2$ be a divisor of~$n$ and $\cU\in\Fqn$ be such that $\delta_s(\cU)=2$ and $\Orb_{\GLhat}(\cU)\subseteq\Orb_N(\cU)$.
Then the orbits coincide and~$\cU$ is one of the subspaces from Examples~\ref{E-SmallNormSingerGLhatOrb}, 
\ref{E-NormSingerGLhatOrb}, and \ref{E-qns=251}.
\end{conjecture}

\section{Isometries of Orbit Codes}\label{S-Isom}

In this section we turn to the question when two orbit codes (under the Singer subgroup or its normalizer) are linearly isometric.
Our first result provides a criterion for when two cyclic orbit codes are not linearly isometric. 

\begin{theo}\label{T-NonIsom}
Let $\cC, \cC'$ be distinct Singer orbits. 
If $\Aut(\cC')=N_{\GL_n(q)}(\Fqn^*) \subseteq \Aut(\cC)$, then $\cC$ and $\cC'$ are not linearly isometric.
\end{theo} 

Our proof is an adaptation of \cite[Thm.~5]{BEOVW16}, where the authors prove an analogous result for $q$-Steiner systems.

\begin{proof}
We prove the contrapositive. 
Suppose that~$\cC$ and~$\cC'$ are linearly isometric, so that there exists $\psi\in\GL_n(q)$ such that 
$\psi(\cC) = \cC'$. 
Let $\tau \in N:=N_{\GL_n}(\Fqn^*)$. 
Then our assumptions on $\Aut(\cC')$ and $\Aut(\cC)$ imply
$\psi\circ\tau\circ\psi^{-1}(\cC') = \psi\circ\tau(\cC) = \psi(\cC) = \cC'$, and thus
$\psi\circ\tau\circ\psi^{-1} \in \Aut(\cC') = N$. 
This shows that~$\psi$ is in the normalizer of~$N$.
But the latter is~$N$ itself thanks to \cref{T-SingerNorm}(a), and hence $\psi\in\Aut(\cC')$ and $\cC=\cC'$.
\end{proof}

The main result of this section shows that Singer orbits of generic subspaces are linearly isometric iff they are Frobenius-isometric. 
This drastically reduces the workload when finding isometry classes of such codes.

\begin{theo}\label{T-IsomCOC}
Let $\cU,\,\cU' \in \Gkn$ such that $1\in\cU'$ and $\cU'$ is generic.
\begin{alphalist}
\item Let $S=\Fqn^*$.
         Then $\Orb_S(\cU)$ and $\Orb_S(\cU')$ are linearly isometric iff they are Frobenius-isometric.
\item Let $k\leq 3n/8$ or $\delta_s(\cU)\geq3$ for all divisors~$s$ of~$n$. 
        Let $N=N_{\GL_n(q)}(\Fqn^*)$. 
        Then $\Orb_N(\cU)$ and $\Orb_N(\cU')$ are linearly isometric iff they are equal.
\end{alphalist}
\end{theo}

\begin{proof}
(a) Only ``$\Longrightarrow$'' needs proof.
Set $\cC=\Orb_S(\cU)$ and $\cC'=\Orb_S(\cU')$.
Let $\psi\in\GL_n(q)$ be such that $\psi(\cC)=\cC'$.
\cref{T-isoOrbCodes}(b) tells us that $\cC'=\Orb_{\psi S\psi^{-1}}(\cU'')$, where $\cU''=\psi(\cU)$.
Hence $\Aut(\cC')$ contains the Singer subgroups~$S$ and $\psi S\psi^{-1}$.
By \cref{T-AutOrb} the automorphism group $\Aut(\cC')$ is contained in $N_{\GL_n(q)}(S)$. 
However, by  \cref{T-SingerNorm}(b)  $N_{\GL_n(q)}(S)$ contains only one Singer subgroup. 
This implies $\psi S\psi^{-1}=S$, and thus
$\psi\in N_{\GL_n(q)}(S)$.
\\
(b) Let $\cC:=\Orb_N(\cU)$ and $\psi(\cC)=\cC':=\Orb_N(\cU')$. Then $\cC'=\Orb_{\psi N\psi^{-1}}(\cU'')$, where $\cU''=\psi(\cU)$, and thus
$\psi N\psi^{-1}\leq \Aut(\cC')= N$, where the last identity follows from \cref{C-AutOrbNormalizer} and \cref{L-NewIneq}.
Hence~$\psi\in N$ thanks to \cref{T-SingerNorm}(a), and thus $\cC'=\psi(\cC)=\cC$.
\end{proof}

As the proof shows, Part~(b) above is true for all subspaces that satisfy $\Aut(\Orb_N(\cU))= N$.
Since the three outliers from Examples~\ref{E-SmallNormSingerGLhatOrb}, \ref{E-NormSingerGLhatOrb}, and 
\ref{E-qns=251} are the only orbit of their size in the respective ambient space, they trivially satisfy the equivalence in~(b) above even though
their automorphism group is much larger.

\medskip
We close this section with some examples and a comparison of isometries and weight-preserving bijections between cyclic orbit codes, 
where we define the weight of a codeword in $\Orb_{\Fqn^*}(\cU)$ as the distance to the `reference space'~$\cU$.
Since we will exclusively consider cyclic orbit codes, we write from now on $\Orb(\cU)$ instead of $\Orb_{\Fqn^*}(\cU)$.

In \cite{GLL19} we studied the weight distribution of cyclic orbit codes $\Orb(\cU)$. 
We will see below that codes with the same weight distribution may not be isometric. 
Before providing details we first summarize the results from~\cite{GLL19}.
Recall the notation from~\eqref{e-ds} and~\eqref{e-distCC}.
As before we assume $k\leq n/2$.

\begin{defi}\label{D-DistDistr}
Let $\cU\in\Gkn$. 
Define $\omega_{i}=|\{\alpha\,\cU\in\Orb(\cU)\mid \alpha \in \Fqn^*, \ds(\cU,\alpha\,\cU)=i\}|$ for 
$i=0,\ldots,2k$. 
We call $(\omega_0,\ldots,\omega_{2k})$ the \emph{weight distribution} of $\Orb(\cU)$.	
\end{defi}

Clearly $\omega_0=1$ and $\omega_i=0$ for $i=1,\ldots,d-1$, where $d=\dd(\Orb(\cU))$.
Obviously, the weight distribution is trivial for spread codes (i.e., if $\dd(\Orb(\cU))=2k$).
From~\eqref{e-ds} it follows that 
$\dd(\Orb(\cU)=2(k-\ell)$, where $\ell=\max\{\dim(\cU\cap\alpha\,\cU)\mid \alpha\in\Fqn^*\}$.

In \cref{T-DistDistr} below we list some facts about the weight distribution.
Part~(a) shows that all cyclic orbit codes with distance $2(k-1)$ have the same weight distribution.
Hence there exists a weight-preserving bijection between any such codes. 
However, as we will see below, the codes are not necessarily isometric.
Subspaces~$\cU$ that generate cyclic orbit codes with distance $2(k-1)$ are known as Sidon spaces; see~\cite{RRT18} where 
also constructions of such spaces can be found.

Not surprisingly, codes with distance up to $2k-4$ do not share the same weight distribution in general.
For distance equal to $2(k-2)$, Part~(b) below provides information about the weight distribution. 
Further details about the parameter~$r$ in Part~(b) can be found in \cite[Sec.~4]{GLL19}. 
However, it is not yet fully understood which values this parameter can assume in general.

\begin{theo}[\mbox{\cite[Thms.~3.7 and~4.1]{GLL19}}]\label{T-DistDistr}
Let $\cU\in\Gkn$ be such that $1\in\cU$. Let $\dd(\Orb(\cU)=2(k-\ell)$, where $\ell>0$. 
Set $Q=(q^k-1)(q^k-q)/(q-1)^2$ and $N=(q^n-1)/(q-1)$.
\begin{alphalist}
\item Suppose $\ell=1$. Then $|\Orb(\cU)|=N$ and 
        \[
            \big(\omega_{2k-2},\,\omega_{2k}\big)=\big(Q,\ N-Q-1\big).
         \]
\item Suppose $\ell=2$ and $|\Orb(\cU)|=N$. Then there exits $r\in\N_0$ and $\epsilon\in\{0,1\}$ such that 
         \[
            \big(\omega_{2k-4},\,\omega_{2k-2},\,\omega_{2k}\big)=\big(\epsilon q+rq(q+1),\ Q-(q+1)\omega_{2k-4},\ 
                 N-\omega_{2k-2}-\omega_{2k-4}-1\big).
         \]
         The case $\epsilon=1$ occurs iff $\cU$ contains the subfield $\F_{q^2}$ (which implies that~$n$ is even).
\end{alphalist}
\end{theo}

In the following examples we list all isometry classes of the subspaces in question along with their automorphism group.
In most cases the size of the isometry class is determined by the automorphism group as follows.

\begin{rem}\label{R-AutGroupIsoClass}
Let~$\cC=\Orb_{\Fqn^*}(\cU)$ be a cyclic orbit code with automorphism group $A$ contained in $N_{\GL_n(q)}(\Fqn^*)$.
Then the isometry class of~$\cC$ consists of~$\nu$ cyclic orbit codes, where $\nu=n(q^n-1)/|A|$.
This is due to \cref{T-IsomCOC}, which tells us that two cyclic orbit codes are isometric iff they belong to the same 
orbit under the normalizer of the Singer subgroup $\Fqn^*$.
\end{rem}

In the following examples, the total number of orbits also follows from the formula for the number of Singer orbits of a given length that is
provided in \cite[Thm.~2.1]{Dru02} for general $(q,n,k)$.

\begin{exa}\label{E-263}
Let $(q,n,k)=(2,6,3)$.
There exist $23$ cyclic orbit codes generated by $3$-dimensional subspaces. 
One of them is $\Orb(\F_{2^3})$, which is a spread code 
(i.e., it consists of~$9$ subspaces and its subspace distance is~$6$; hence the union of its subspaces is~$\F_{2^6}$). 
Its automorphism group is $\Aut(\Orb(\F_{2^3}))=N_{\GL_6(2)}( \GL_2(2^3))$.
This follows directly from \cref{T-AutOrb} along with the fact that $\Gal(\F_{2^3}\mmid\F_2)$ acts trivially on $\Orb(\F_{2^3})$.
Clearly, this is the only orbit generated by a non-generic subspace of $\F_{2^6}$.
Even more, it is the only orbit with a generating subspace~$\cU$ such that $\delta_3(\cU)\neq2$ (see \cref{D-Uhat}).  
The other $22$ orbits have length $2^6-1$, and their automorphism group is contained in $N_{\GL_6(2)}(\F_{2^6}^*)$ thanks to \cref{T-AutOrb}.
They classify as follows.
Note that distance~$4$ corresponds to Case (a) of the above theorem and distance~$2$ to Case~(b). 
In the latter case we also present the value of $\omega_{2k-4}=\omega_2$ (which fully determines the weight distribution).
It is, of course, invariant under isometry and thus identical for all orbits in the isometry class. 
Finally, we also present~$\delta_2(\cU)$ for any subspace~$\cU$ in any of the orbits.
\begin{alphalist}
\item Orbits with automorphism group~$\F_{2^6}^*$:
        \\
        -- 1 isometry class, consisting of orbits with distance 4 ($\delta_2(\cU)=3$).
        \\
        -- 1 isometry class, consisting of orbits with distance 2 and $\omega_2=6$ ($\delta_2(\cU)=3$). 
\item  Orbits with automorphism group $\Gal(\F_{2^6}\mmid\F_{2^3})\rtimes \F_{2^6}^*$:
         \\
        -- 1 isometry class, consisting of orbits with distance~$2$ and $\omega_2=2$ ($\delta_2(\cU)=2$). 
        \\
         -- 1 isometry class, consisting of orbits with distance~$2$ and $\omega_2=6$ ($\delta_2(\cU)=3$). 
\item  Orbits with automorphism group $\Gal(\F_{2^6}\mmid\F_{2^2})\rtimes \F_{2^6}^*$:
         \\
         -- 1 isometry class, consisting of orbits with distance~$4$ ($\delta_2(\cU)=3$).
         \\
         -- 1 isometry class, consisting of orbits with distance~$2$ and $\omega_2=2$ ($\delta_2(\cU)=2$). 
\end{alphalist}
\end{exa}

\begin{exa}\label{E-273}
Let $(q,n,k)=(2,7,3)$.
In this case, there are no proper subfields of $\F_{2^7}$ to be taken into account, and in particular $\epsilon=0$
in Case~(b) of \cref{T-DistDistr}.
There exist $93$ cyclic orbit codes generated by $3$-dimensional subspaces. 
All of them have length $2^7-1$.
They classify  as follows.
\begin{alphalist}
\item Orbits with automorphism group~$\F_{2^7}^*$:
        \\
        -- 10 isometry classes, consisting of orbits with distance 4.
        \\
        -- 3 isometry classes, consisting of orbits with distance 2 and $\omega_2=6$.
\item  Orbits with automorphism group $\Gal(\F_{2^7}\mmid\F_2)\rtimes \F_{2^7}^*$:
         \\
         --  2 isometry classes, each consisting of a single orbit with distance~$4$.
\end{alphalist}
\end{exa}

\begin{exa}\label{E-283}
Let $(q,n,k)=(2,8,3)$. 
There exist $381$ cyclic orbit codes generated by a $3$-dimensional subspace. All orbits have length $2^8-1$.
Exactly one orbit is generated by a subspace contained in~$\F_{2^4}$.
Clearly, all other orbits are generated by subspaces~$\cU$ with $\delta_4(\cU)=2$.
The orbits classify as follows. We present the data as in \cref{E-263}.
\begin{alphalist}
\item Orbits with automorphism group~$\F_{2^8}^*$:
        \\
        -- 38 isometry classes, consisting of orbits with distance 4 ($\delta_2(\cU)=3$).
        \\
        -- 4 isometry classes, consisting of orbits with distance 2 and $\omega_2=6$ ($\delta_2(\cU)=3$).
        \\
        -- 2 isometry classes, consisting of orbits with distance 2 and $\omega_2=2$ ($\delta_2(\cU)=2$).  
\item  Orbits with automorphism group $\Gal(\F_{2^8}\mmid\F_{2^4})\rtimes \F_{2^8}^*$:
        \\
        -- 3 isometry classes, consisting of orbits with distance 4 ($\delta_2(\cU)=3$).
        \\
        -- 2 isometry classes, consisting of orbits with distance 2 and $\omega_2=6$ ($\delta_2(\cU)=3$).
        \\
        -- 1 isometry class, consisting of orbits with distance 2 and $\omega_2=2$ ($\delta_2(\cU)=2$).
   
\item  Orbits with automorphism group $\Gal(\F_{2^8}\mmid\F_{2^2})\rtimes \F_{2^8}^*$:
        \\
        -- 2 isometry classes, consisting of orbits with distance 4 ($\delta_2(\cU)=3$).
\item  Orbits with automorphism group $\Gal(\F_{2^4}\mmid\F_{2})\rtimes\GL_2(2^4)$:
         \\
        -- 1 isometry class, consisting of a single orbit with distance 2 and $\omega_2=14$ ($\delta_2(\cU)=2$).
        This cyclic orbit code is the only orbit generated by a subspace contained in $\F_{2^4}$ (and it contains $\F_{2^2}$).
\end{alphalist}
\end{exa}

\section*{Conclusion and Open Problems}
We studied orbits of $\F_q$-subspaces of~$\Fqn$ under the Singer subgroup and under the normalizer of the Singer group.
For cyclic orbit codes generated by generic subspaces we proved that a linear isometry between such orbits
is contained in the normalizer of the Singer group.
The result implies that, for most parameter cases, distinct orbits under the normalizer of the Singer subgroup  are not linearly isometric. 
The following questions remain.
\begin{alphalist}
\item We strongly believe that the isometry result for orbits under the normalizer is true for all parameter cases.
This would follow if \cref{C-Conjec} can be established, that is: the automorphism group of a normalizer orbit generated by a 
subspace~$\cU$ does not contain the field-extension subgroup $\GLhat$ if~$\cU$ is not contained in~$\Fqs$ -- unless~$\cU$ is one of the exceptional cases from Examples~\ref{E-SmallNormSingerGLhatOrb}, \ref{E-NormSingerGLhatOrb}, and \ref{E-qns=251}.
\item Furthermore, our isometry result in \cref{T-IsomCOC}  is true only for orbits generated by generic subspaces. 
It is an open question whether the same result is true for arbitrary orbits.
\item Finally, as we briefly address in Section~\ref{S-AutGroupSingerNorm} we believe that any subspace $\cU\subseteq\Fqn$ satisfies
$\Orb_N(\cU)\subseteq\Orb_{\GLhat}(\cU)$, where $N=N_{\GL_n(q)}(\Fqn^*)$ and $s\leq n/2$ is any divisor of~$n$.
We have to leave this to future research.
\end{alphalist}

\section*{Appendix}
\renewcommand{\theequation}{A.\arabic{equation}}

\noindent\emph{Proof of \cref{T-AdjGroups}:}
Let~$\omega$ be a primitive element of~$\Fqn$ and $f=X^n-\sum_{i=0}^{n-1}f_i X^i\in\F_q[X]$ be its minimal polynomial.
We proceed in several steps.
\\
\underline{Step 1:} Recall the maps $m_a$ from~\eqref{e-ma}. According to \cref{D-GLhat} we identify $\Fqn^*$ with 
\[
            \Fqn^*=\subgroup{m_{\omega}}=\{m_{\omega^i}\mid i=0,\ldots,q^n-2\}. 
\]
We determine all maps $\rho\in\GL_n(q)$ such that $\rho^{-1}\circ m_\omega^\dagger\circ\rho=m_\omega$.
These maps then clearly satisfy $\rho^{-1}(\Fqn^*)^\dagger\rho=\Fqn^*$, which is what we want.
The reader may recall the fact that any matrix~$A$ is similar to its transpose $A\T$
(use for instance the fact that they share the same invariant factors).
Hence there exists at least one such map $\rho\in\GL_n(q)$. 
However, we need to determine all of them explicitly in order to select a suitable one in a later step.
Consider the recurrence relation
\begin{equation}\label{e-recurrence}
   x_{j+n}=\sum_{i=0}^{n-1}f_ix_{j+i}\ \text{ for }\ j\geq0.
\end{equation}
For every initial condition $x_0=a_0,\ldots,x_{n-1}=a_{n-1}$ the recurrence~\eqref{e-recurrence} has a unique solution, 
which we denote by $(a_i)_{i\in\N_0}$. 
Set $\cR:=\{\rho\in\GL_n(q)\mid \rho^{-1}\circ m_\omega^\dagger\circ\rho=m_\omega\}$. 
Thus every $\rho\in\cR$ satisfies $\rho^{-1}(\Fqn^*)^\dagger\rho=\Fqn^*$.
We show 
\begin{equation}\label{e-rhoset}
   \cR =\bigg\{\rho\in\text{End}_{\F_q}(\Fqn)\bigg|\begin{array}{l}\exists (a_0,\ldots,a_{n-1})\in\Fqn\setminus0:\\[.6ex]
                \rho(\omega^i)=\sum_{j=0}^{n-1}a_{j+i}\omega^j\ \text{ for }\ i\in\N_0\end{array}\bigg\}.
\end{equation}
Note that this identity tells us that the maps~$\rho\in\cR$ are fully determined by the value of~$\rho(1)$, which is given as $\sum_{j=0}^{n-1}a_j\omega^j$.
\\
`$\subseteq$'
Let $\rho\in\cR$.
Set $\rho(1)=a\in\Fqn^*$ and write $a=\sum_{j=0}^{n-1}a_j\omega^j$.
Since $\inner{\omega^i}{\omega^j}=\delta_{i,j}$ for $i,j=0,\ldots,n-1$ this means $\inner{\rho(1)}{\omega^j}=a_j$ for $j=0,\ldots,n-1$.
Inducting on~$i$ we show now that
\begin{equation}\label{e-rho}
    \rho(\omega^i)=\sum_{j=0}^{n-1}a_{j+i}\omega^j\ \text{ for all }\ i\in\N_0. 
\end{equation}
It is clearly true for $i=0$.
For the induction step note first that the identity $m_\omega^\dagger\circ\rho=\rho\circ m_\omega$ is equivalent to 
\begin{equation}\label{e-omegarhoadj}
    \inner{\rho(\omega y)}{z}=\inner{\rho(y)}{\omega z} \text{ for all }y,z\in\Fqn.
\end{equation}
Assuming now~\eqref{e-rho} and using~\eqref{e-omegarhoadj} we obtain
$\inner{\rho(\omega^{i+1})}{\omega^j}=\inner{\rho(\omega^{i})}{\omega^{j+1}}=a_{j+1+i}$ for $j=0,\ldots,n-2$ and 
$\inner{\rho(\omega^{i+1})}{\omega^{n-1}}=\inner{\rho(\omega^{i})}{\omega^{n}}=\sum_{j=0}^{n-1}f_j\inner{\rho(\omega^{i})}{\omega^{j}}
     =\sum_{j=0}^{n-1}f_ja_{j+i}=a_{j+n}$,
where the last step follows from~\eqref{e-recurrence}.
Hence $\rho(\omega^{i+1})=\sum_{j=0}^{n-1}a_{j+i+1}\omega^j$, and this establishes~\eqref{e-rho}.
All of this shows that~$\rho$ is in the set on the right hand side of~\eqref{e-rhoset}.
\\
`$\supseteq$' Let~$\rho$ in the set on the right hand side of~\eqref{e-rhoset}.
In order to establish~\eqref{e-omegarhoadj} it suffices to show that
\begin{equation}\label{e-omegaij}
  \inner{\rho(\omega^{i+1})}{\omega^j}=\inner{\rho(\omega^i)}{\omega^{j+1}}\ \text{ for all }\ i,j=0,\ldots,n-1.
\end{equation}
The left hand side simplifies to 
$\inner{\rho(\omega^{i+1})}{\omega^j}=\sum_{\ell=0}^{n-1}a_{\ell+i+1}\inner{\omega^\ell}{\omega^j}=a_{j+i+1}$ for all $j=0,\ldots,n-1$.
For $j=0,\ldots,n-2$ the right hand side of~\eqref{e-omegaij} turns into
$\inner{\rho(\omega^i)}{\omega^{j+1}}=\sum_{\ell=0}^{n-1}a_{\ell+i}\inner{\omega^\ell}{\omega^{j+1}}=a_{j+1+i}$, while for $j=n-1$ we have
\[
   \inner{\rho(\omega^i)}{\omega^{j+1}}= \inner{\rho(\omega^i)}{\omega^{n}}=\sum_{\ell=0}^{n-1}a_{\ell+i}\sum_{r=0}^{n-1}f_r\inner{\omega^\ell}{\omega^r}
    =\sum_{\ell=0}^{n-1}a_{\ell+i}f_\ell=a_{n+i}=a_{j+1+i},
\]
where the penultimate identity follows from~\eqref{e-recurrence}.
All of this establishes~\eqref{e-omegaij}.
In order to complete the proof of~\eqref{e-rhoset} it remains to show that~$\rho$ is an isomorphism. 
Assume $\rho(b)=0$ for some $b\in\Fqn$. 
Then~\eqref{e-omegarhoadj} implies
\[
    0=\inner{\rho(b)}{z}=\inner{\rho(\omega b)}{\omega^{-1} z}\ \text{ for all }\ z\in\Fqn.
\]
Hence $\rho(\omega b)=0$ and thus $\rho(\omega^ib)=0$ for all~$i=0,\ldots,q^n-2$.
Since $\rho(1)=a\neq0$, this implies $b=0$. Thus~$\rho$ is injective and an isomorphism.
\\[.6ex]
\underline{Step 2:} Let $\sigma:\Fqn\longrightarrow\Fqn$ be the Frobenius homomorphism, thus $\sigma(z)=z^q$ for all $z\in\Fqn$.
We now want to determine a map~$\rho\in\cR$ satisfying $\rho^{-1}\circ\sigma^\dagger\circ\rho=\sigma^{-1}$.
Consider the $\F_q$-linear map $\xi:\Fqn\longrightarrow\Fqn,\ z\longmapsto \sigma(z)-z$. Clearly, $\ker\xi=\F_q$.
Thus, $\dim(\im\xi)=n-1$.
Pick now 
\[
       \rho(1)\in(\im\xi)^\perp\setminus0\ \text{(which is unique up to $\F_q$-scalar multiples).}
\]
Thanks to~\eqref{e-rhoset} this determines a  unique map~$\rho\in\cR$. 
The choice of $\rho(1)$ implies
\begin{equation}\label{e-rho1}
        \inner{\rho(1)}{\sigma(z)}=\inner{\rho(1)}{z} \ \text{ for all }z\in\Fqn.
\end{equation}
With the aid of~\eqref{e-omegarhoadj} we obtain
\[
   \inner{\rho(\omega^i)}{\omega^j}=\inner{\rho(1)}{\omega^{i+j}}=\inner{\rho(1)}{\omega^{(i+j)q}}
   =\inner{\rho(\omega^{iq})}{\omega^{jq}}\ \text{ for all }\ i,j\in\N_0.
\]
Since $1,\omega,\ldots,\omega^{n-1}$ is an $\F_q$-basis of $\Fqn$, this implies $\inner{\rho(y)}{z}=\inner{\rho(\sigma(y))}{\sigma(z)}$ 
for all $z,y\in\Fqn$.
The latter is equivalent to $\inner{\rho(\sigma^{-1}(y)}{z}=\inner{\rho(y)}{\sigma(z)}$ for all $z,y\in\Fqn$, and this means
$\rho\circ\sigma^{-1}=\sigma^\dagger\circ\rho$. 
All of this implies $\rho^{-1}\Gal(\Fqn\mmid\F_q)^\dagger\rho=\Gal(\Fqn\mmid\F_q)$.
\\[.6ex]
\underline{Step 3:} Let $s$ be a divisor of~$n$ and consider the subgroup $\GLhat$ of $\GL_n(q)$. Let $\gamma\in\GLhat$. 
We have to show that $\rho^{-1}\circ\gamma^\dagger\circ\rho=:\hat{\gamma}$ is in $\GLhat$, i.e., that $\hat{\gamma}$ is $\Fqs$-linear.
Let $N=(q^n-1)/(q^s-1)$. 
Then $\Fqs=\F_q[\omega^N]$ and thus it suffices to prove that
\begin{equation}\label{e-gammahat}
      \hat{\gamma}(\omega^N y)=\omega^N\hat{\gamma}(y)\ \text{ for all }\ y\in\Fqn.
\end{equation}
Using $\rho\circ\hat{\gamma}=\gamma^\dagger\circ\rho$ together with~\eqref{e-omegarhoadj} and the $\Fqs$-linearity of~$\gamma$, 
we compute for $y,z\in\Fqn$
\begin{align*}
      \inner{\rho\circ\hat{\gamma}(\omega^N y)}{z}&=\inner{\rho(\omega^N y)}{\gamma(z)}=\inner{\rho(y)}{\omega^N\gamma(z)}
                    =\inner{\rho(y)}{\gamma(\omega^N z)}\\
          &=\inner{\gamma^\dagger(\rho(y))}{\omega^N z}=\inner{\rho(\hat{\gamma}(y))}{\omega^N z}
                   =\inner{\rho(\omega^N\hat{\gamma}(y))}{z}.
\end{align*} 
Since this is true for all $z\in\Fqn$ and since~$\rho$ is an isomorphism, this implies~\eqref{e-gammahat}.
All of this proves $\rho^{-1}\GLhat^\dagger\rho^{-1}=\GLhat$.
\hfill $\square$

\bibliographystyle{abbrv}

\end{document}